\documentclass[11pt,twoside,letterpaper]{article}

\usepackage{amsmath}
\usepackage{amsthm}
\usepackage{amssymb}
\numberwithin{equation}{section}
\usepackage{color}
\theoremstyle{plain}
\newtheorem{proposition}{Proposition}
\newtheorem{corollary}[proposition]{Corollary}
\newtheorem{lemma}[proposition]{Lemma}
\newtheorem{theorem}[proposition]{Theorem}

\theoremstyle{definition}

\newtheorem{example}[proposition]{Example}
\newtheorem{remark}[proposition]{Remark}

\DeclareMathOperator{\diag}{diag}


\makeatletter
\def\cleardoublepage{\clearpage\if@twoside \ifodd\c@page\else%
    \hbox{}%
    \thispagestyle{empty}%
    \newpage%
    \if@twocolumn\hbox{}\newpage\fi\fi\fi}
\makeatother

\def\figurename{Figure}
\makeatletter
\renewcommand{\fnum@figure}[1]{\figurename~\thefigure.}
\makeatother

\def\tablename{Table}
\makeatletter
\renewcommand{\fnum@table}[1]{\tablename~\thetable.}
\makeatother

\begin{document}

\title{
{
\vskip 0.45in
\bfseries\scshape Inverse problems associated with integrable equations of Camassa-Holm type; explicit formulas on the real axis, I}}
\author{\bfseries\itshape Keivan Mohajer \thanks{Department of Mathematics,  University of Isfahan, Isfahan, 81746-73441, Iran; k.mohajer@sci.ui.ac.ir}
\and
\bfseries\itshape Jacek Szmigielski\thanks{Department of Mathematics and Statistics, University of Saskatchewan, 106 Wiggins Road, Saskatoon, Saskatchewan, S7N 5E6, Canada; szmigiel@math.usask.ca}
}

\date{}

\maketitle
\begin{abstract}
  The inverse problem which arises in the
  Camassa--Holm equation
  is revisited for the class of discrete densities.  The method of solution relies
  on the use of orthogonal polynomials.  The explicit formulas are
  obtained directly from the analysis on the real axis without any additional transformation
  to a ``string'' type boundary value problem known from prior works.
  \end{abstract}
\vspace{0.1 in}
\noindent {\bf Hongyou Wu in memoriam}
\vspace{.1 in}

\vspace{.08in} \noindent \textbf {AMS Subject Classification:} Primary 37K15.

\section{Introduction}
The purpose of this paper and its sequel is to present an alternative solution to two inverse problems
which appear in the context of Camassa-Holm type equations.  We concentrate in this paper on the Camassa-Holm equation (CH)\cite{camassa-holm}:
\begin{equation}
  \label{eq:CH}
  u_t - u_{xxt} + 3u u_x = 2u_x u_{xx} + u u_{xxx},
\end{equation}
while in the sequel we study the inverse problem for another equation, discovered by V. Novikov (VN) \cite{novikov}:
\begin{equation}
  \label{eq:novikov}
  u_t - u_{xxt} + 4 u^2 u_x = 3 u u_x u_{xx} + u^2 u_{xxx}.
\end{equation}
Both these equations admit another, often used, representation
\begin{equation}\label{eq:CH-system}
m_t+m_xu+2mu_x=0, \qquad m=u-u_{xx},
\end{equation}
for the CH equation, and \begin{equation}
  \label{eq:novikov-system}
  m_t + (m_x u + 3 m u_x) \, u = 0, \qquad m = u - u_{xx},
\end{equation}
for the VN equation.   The solution to the inverse problem presented in this paper
is obtained directly on the real axis, instead of transforming the problem to the
inhomogeneous string problem as was done in \cite{beals-sattinger-szmigielski-moment}
and adapting the method used by T. Stieltjes in \cite{stieltjes}.

We give a self-contained presentation of the inverse problem for Equation \eqref{eq:BVP},
emphasizing the role of orthogonal polynomials (Proposition \ref{prop:orthogonality}) and
culminating in the elegant inverse formulas stated in Corollary \ref{cor:xm}.
In the sequel, for the case of the VN equation \eqref{eq:novikov}, we will show
that a different class of polynomials is germane to the inverse problem, namely a
class of biorthogonal polynomials (Cauchy biorthogonal polynomials) studied  in \cite{BGS1}.

\section{CH inspired inverse problem}

 The well known Lax pair system for the CH equation can be written as
 follows (see \cite{camassa-holm})
\begin{equation}\label{CH-system}
\begin{split}
&\bigl(1-D_x^2-zm\bigr)\psi=0, \\
&\bigl(D_t+(u+1/z)D_x-u_x/2-1/z \bigr)\psi=0.
\end{split}
\end{equation}
We would like to point out that this Lax pair yields a slightly
different normalization than in \eqref{eq:CH-system}, here,
$m=2u-\frac12 u_{xx}$. If we set
$m(x,t)=\sum_{j=1}^{n}m_j(t)\delta_{x_j}$ where $m_j(t)$ and
$x_j(t)$ are $C^{\infty}$ functions, we obtain what is called the
$n$-peakon Lax pair for the CH equation. We assume that for every
fixed $t\geq 0$ the function $\psi(x,t)$ is continuous on
$\mathbb{R}$ and $x_1(0)<x_2(0)<\dots < x_n(0)$. Now the first
equation of the system \eqref{CH-system} implies that on the
interval $(x_j,x_{j+1})$, we have
\begin{equation}
\psi=A_j e^x + B_j e^{-x},
\end{equation}
It is well known that requiring $\psi\rightarrow 0$
as $|x| \rightarrow \infty$ is consistent with the system
\eqref{CH-system}.   Let's
therefore assume that $A_0=1$ and $B_0=0$. Then from the
continuity of $\psi$ and the jump conditions imposed by the first
equation of \eqref{CH-system} we have
\begin{equation}
\begin{pmatrix}
A_j\\B_j
\end{pmatrix}
=K_j(z)\begin{pmatrix} 1\\0 \end{pmatrix},\ \ \ 1\leq j \leq n
\end{equation}
where
\begin{equation}
K_j(z)=T_j(z)T_{j-1}(z)\dots T_1(z),
\end{equation}
and
\begin{equation}
T_k(z)=
\begin{pmatrix}
1-\frac{z}{2}m_k & -\frac{z}{2}m_ke^{-2x_k}\\
\frac{z}{2}m_ke^{2x_k} & 1+\frac{z}{2}m_k
\end{pmatrix},\ \ \ \ 1\leq k \leq n.
\end{equation}
Clearly,
\begin{theorem}
The spectrum of the boundary value problem
\begin{equation}\label{eq:BVP}
\bigl(1-D_x^2-zm\bigr)\psi=0 , \qquad \psi \rightarrow 0, \text{ for $|x|\rightarrow \infty$ }
\end{equation}
is given by roots of $A_n(z)=0$.
\end{theorem}
In order to get more information about the spectrum we can rewrite equation
\eqref{eq:BVP} as an integral equation
\begin{equation*}
\psi(x,z)=z\int K(x,y) \psi(y,z)dM(x),
\end{equation*}
where $M$ is the distribution function of the measure $m(x)$ and $K(x,y)=\frac12 e^{-|x-y|}$ is the
unique Greens function of $1-D^2_x$ vanishing at $\pm \infty$.  One can easily write
the explicit form of this integral equation as:
\begin{equation} \label{eq:IE}
\psi(x,z)=z\sum_{j=1}^n \frac12 e^{-|x-x_j|} \psi(x_j,z)m_j,
\end{equation}
where, by assumption, all $m_j>0$.
Moreover, evaluating at $x_i$, and introducing the matrix notation
\begin{equation}\label{eq:E}
\pmb{\psi }=(\psi(x_1,z),\psi(x_2,z),\dots, \psi(x_n,z))^T, \qquad \mathbf{E}=\big[\frac12 e^{-|x_i-x_j|}\big], \end{equation}
 and
 \begin{equation}\label{eq:P}
 \mathbf{P}=\diag(m_1,m_2,\dots, m_n), \end{equation}
we obtain  equation \eqref{eq:IE} in the form of the matrix
eigenvalue problem:
\begin{equation*}
\pmb{\psi}=z\mathbf E \mathbf P \, \pmb{\psi}.
\end{equation*}
$\mathbf E$ is an example of a single-pair matrix (see \cite{gantmacher-krein}) which is oscillatory and so is $\mathbf E \mathbf P$.  Oscillatory matrices form a subset of {\sl totally nonnegative matrices}, which can
be characterized as matrices whose all minors of arbitrary
size are nonnegative.  Oscilatory matrices are special in that they have simple, strictly positive,
spectrum \cite{gantmacher-krein}.
We therefore have the following description of the spectrum of the boundary value problem
\eqref{eq:BVP}.
\begin{lemma}\label{lem:spectrum}
$A_n(z)=\det (I-z\mathbf E \mathbf P)={\displaystyle\prod_{1\leq k\leq n} \big(1-\frac{z}{\lambda_k}\big)} $ where $0<\lambda_1<\lambda_2<\dots <\lambda_n$
\end{lemma}

We need one more fact.
\begin{theorem}\label{thm:ST}
The Weyl function $W(z)=-\frac{B_n(z)}{A_n(z)}$ is a (shifted) Stieltjes transform of a measure $d\mu(x)=\sum_{j=1}^n
a_j \delta_{\lambda_j}, \, a_j>0$, that is
\begin{equation*}
W(z)=\gamma +\int \frac{1}{z-x}d\mu(x),
\end{equation*}
where $\gamma=\int \frac{1}{x}d\mu(x)$.

\end{theorem}
\begin{proof}
This theorem follows, after a slight adjustment of notation, from Theorem 4.2 proven in \cite{beals-sattinger-szmigielski-moment}
and the observation that $W(0)=0$ which is a consequence of $T_k(0)=I$.
\end{proof}
The continued fraction expansion (\cite{stieltjes}) turns out to be at the heart of the inverse problem, in a full
analogy with the inverse problem for an inhomogeneous string which was studied
in the 1950's by M.G. Krein (\cite{Kreinstring}, \cite{dym-mckean-gaussian}).
\begin{theorem}
\begin{equation}
\begin{split}
&\frac{B_n}{A_n}=-e^{2x_n}+\\
&\cfrac{1}{-\frac{z}{2}m_ne^{-2x_n}+\cfrac{1}{e^{2x_n}-e^{2x_{n-1}}+\cfrac{1}{-\frac{z}{2}m_{n-1}e^{-2x_{n-1}}+\cfrac{1}{{\ddots}\
\ \
_{+\cfrac{1}{-\frac{z}{2}m_1e^{-2x_1}+\cfrac{1}{e^{2x_1}}}}}}}},
\end{split}
\end{equation}
\end{theorem}
\begin{proof}
Since
\begin{equation}
\begin{pmatrix} A_n\\B_n \end{pmatrix}=T_n\begin{pmatrix} A_{n-1}\\B_{n-1} \end{pmatrix},
\end{equation}
we have
\begin{equation}
\frac{B_n}{A_n}=\frac{\frac{z}{2}m_ne^{2x_n}A_{n-1}+(1+\frac{z}{2}m_n)B_{n-1}}{(1-\frac{z}{2}m_n)A_{n-1}-\frac{z}{2}m_ne^{-2x_n}B_{n-1}}.
\end{equation}
This equation is equivalent to
\begin{equation}
\frac{B_n}{A_n}=-e^{2x_n}+\cfrac{1}{-\frac{z}{2}m_ne^{-2x_n}+\cfrac{1}{e^{2x_n}+\cfrac{B_{n-1}}{A_{n-1}}}}.
\end{equation}
Hence, the proof is complete by induction.
\end{proof}
Suppose that $P_j/Q_j$ $(0\leq j \leq 2n)$ is the $j$th convergent
of the negative of the continued fraction we obtained above so that we are
approximating the Weyl function $W(z)$. Then we observe that
\[
Q_0=1,\ \ \ P_0=e^{2x_n},
\]
\begin{equation}
\frac{P_1}{Q_1}=\frac{1+\frac{z}{2}m_n}{\frac{z}{2}m_n
e^{-2x_n}},
\end{equation}
and
\begin{equation}
\frac{P_2}{Q_2}=\frac{e^{2x_{n-1}}-\frac{z}{2}m_n(e^{2x_n}-e^{2x_{n-1}})}{1-\frac{z}{2}m_n(1-e^{2(x_{n-1}-x_n)})}.
\end{equation}
Therefore, it can be verified that
\begin{equation}
\begin{split}
&Q_2=-(e^{2x_n}-e^{2x_{n-1}})Q_1+Q_0,\\
&P_2=-(e^{2x_n}-e^{2x_{n-1}})P_1+P_0,
\end{split}
\end{equation}
\begin{equation}
\begin{split}
&Q_3=\frac{z}{2}m_{n-1}e^{-2x_{n-1}}Q_2+Q_1,\\
&P_3=\frac{z}{2}m_{n-1}e^{-2x_{n-1}}P_2+P_1.
\end{split}
\end{equation}
These equations guide us to the following theorem
\begin{theorem}
\begin{equation}
\begin{split}
&Q_{2k}=-(e^{2x_{n-k+1}}-e^{2x_{n-k}})Q_{2k-1}+Q_{2k-2},\\
&P_{2k}=-(e^{2x_{n-k+1}}-e^{2x_{n-k}})P_{2k-1}+P_{2k-2},
\end{split}
\end{equation}
\begin{equation}
\begin{split}
&Q_{2k+1}=\frac{z}{2}m_{n-k}e^{-2x_{n-k}}Q_{2k}+Q_{2k-1},\\
&P_{2k+1}=\frac{z}{2}m_{n-k}e^{-2x_{n-k}}P_{2k}+P_{2k-1}.
\end{split}
\end{equation}
\end{theorem}
\begin{proof}
By induction.
\end{proof}
\begin{corollary}
\begin{equation}
P_jQ_{j-1}-P_{j-1}Q_j=(-1)^{j-1},\ \ \ \ 1\leq j \leq 2n.
\end{equation}
\end{corollary}
\begin{proof}
Use induction and the previous theorem.
\end{proof}
\begin{corollary}\label{cor:degrees}
\begin{equation}
\begin{split}
& \deg(P_{2k+1})=\deg(Q_{2k+1})=k+1,\\
& \deg(P_{2k})=\deg(Q_{2k})=k.
\end{split}
\end{equation}
\end{corollary}
In preparation for the inverse problem we introduce the following notation.

\noindent {\bf Notation:} Given a polynomial $p(z)$ we denote by
$p[j]$ the coefficient at the $j$th power of $z$.  Moreover, $n-k$ is
abbreviated to $k'$.

With this notation in place we have
\begin{theorem}\label{thm:formulasIP}

\begin{align*}
&Q_{2k}[0]=1, \quad P_{2k}[0]=e^{2x_{k'}},\\
&Q_{2k+1}[0]=0, \quad P_{2k+1}[0]=1, \qquad P_{2k+1}[1]=\frac12 \sum_{j=0}^km_{j'}
\end{align*}

\end{theorem}
\begin{proof} By induction. \end{proof}

The Weyl function of the $n$-peakon problem reads:
\begin{proposition}
\begin{equation*}
W(z)=\frac{P_{2n}(z)}{Q_{2n}(z)}.
\end{equation*}
\end{proposition}
The inverse problem hinges on the following approximation formulas:
\begin{theorem}\label{thm:approx}
\begin{equation}
\begin{split}
&
W(z)-\frac{P_{2k+1}(z)}{Q_{2k+1}(z)}=\mathcal{O}(\frac{1}{z^{2k+1}}),\
\ \ \ k=0,\dots,n-1,\ \ \ z\to\infty,\\
&
W(z)-\frac{P_{2k}(z)}{Q_{2k}(z)}=\mathcal{O}(\frac{1}{z^{2k+1}}),\
\ \ \ k=0,\dots,n-1,\ \ \ z\to\infty.
\end{split}
\end{equation}
\end{theorem}
\begin{proof}
By the previous theorem we have
\begin{equation*}
\frac{P_{2k+1}}{Q_{2k+1}}-\frac{P_{2k}}{Q_{2k}}=\frac{1}{Q_{2k+1}Q_{2k}}=\mathcal{O}(\frac{1}{z^{2k+1}}),\
\ \ \ z\to\infty.
\end{equation*}
Also,
\begin{equation*}
\frac{P_{2k}}{Q_{2k}}-\frac{P_{2k-1}}{Q_{2k-1}}=\frac{-1}{Q_{2k}Q_{2k-1}}=\mathcal{O}(\frac{1}{z^{2k-1}}),\
\ \ \ z\to\infty.
\end{equation*}
From these and the definition of the Weyl function we get
\begin{equation*}
W(z)-\frac{P_{2k-1}}{Q_{2k-1}}=\frac{P_{2n}(z)}{Q_{2n}(z)}-\frac{P_{2k-1}}{Q_{2k-1}}=\mathcal{O}(\frac{1}{z^{2k-1}}),\
\ \ \ z\to\infty.
\end{equation*}
Similarly, we see that
\begin{equation*}
W(z)-\frac{P_{2k}}{Q_{2k}}=\mathcal{O}(\frac{1}{z^{2k+1}}),\ \ \
z\to\infty.
\end{equation*}
\end{proof}
Since $Q_{2k+1}(0)=0$ we set ${Q_{2k+1}(z)=z\,q_{2k+1}(z), \,
\deg\, q_{2k+1}=k.}$

\begin{proposition}(orthogonality)\label{prop:orthogonality}
\begin{subequations}
\begin{align}
&\int x^j Q_{2k}(x) d \mu(x)=0,\ \ j=0,\dots,k-1,\ \ k=0,\dots,n-1\\
&\int q_{2k+1}(x) d \mu(x)=1,\ \ \ k=0,\dots,n-1, \\
&\int x^j q_{2k+1}(x) d \mu(x)=0,\ \ j=1,\dots,k-1,\ \
k=0,\dots,n-1.
\end{align}
\end{subequations}

\end{proposition}
\begin{proof}
By the first approximation formula of Theorem \ref{thm:approx} we
have
\begin{equation}
W(z)Q_{2k+1}(z)-P_{2k+1}(z)=\mathcal{O}(\frac{1}{z^k}),\ \ \ \ \
k=0,\dots,n-1.
\end{equation}
Therefore,
\begin{equation*}
z^j(W(z)Q_{2k+1}(z)-P_{2k+1}(z))=\mathcal{O}(\frac{1}{z^{k-j}}),\
\ \ j=0,\dots,k-2.
\end{equation*}
Since $W(z)$ is analytic at $z=\infty$ so is the remainder and its
order there is $\mathcal{O}(\frac{1}{z^2})$.  We can therefore
use the residue calculus.  To this end we choose a circle $\Gamma$ whose radius is large enough to contain the support of
$\mu$ defined in Theorem \ref{thm:ST}.   An elementary contour integration implies then
\begin{equation*}
\frac{1}{2\pi i}\int_\Gamma\int \frac{z^j Q_{2k+1}(z)}{z-x}
d\mu(x) dz= 0,\ \ \ \ j=0,\dots,k-2.
\end{equation*}
Consequently, applying Fubini's theorem and Cauchy's residue theorem we
have
\begin{equation*}
\int x^j Q_{2k+1}(x) d \mu(x)=0,\ \ j=0,\dots,k-2,\ \
k=0,\dots,n-1.
\end{equation*}
Similarly, from the second approximation formula of Theorem
\ref{thm:approx} we have
\begin{equation}
W(z)Q_{2k}(z)-P_{2k}(z)=\mathcal{O}(\frac{1}{z^{k+1}}),\ \ \ \ \
k=0,\dots,n-1.
\end{equation}
Thus,
\begin{equation*}
\int x^j Q_{2k}(x) d \mu(x)=0,\ \ j=0,\dots,k-1,\ \ k=0,\dots,n-1.
\end{equation*}
Also we have
\begin{equation}
\frac{W(z)Q_{2k+1}(z)}{z}-\frac{P_{2k+1}(z)}{z}=\mathcal{O}(\frac{1}{z^{k+1}}),\
\ \ \ \ k=0,\dots,n-1.
\end{equation}
Since $Q_{2k+1}[0]=0$ and $P_{2k+1}[0]=1$, we have
\begin{equation*}
\frac{1}{2\pi i}\int_\Gamma\int \frac{z^{-1} Q_{2k+1}(z)}{z-x}
d\mu(x) dz-1= 0.
\end{equation*}
Hence,
\begin{equation*}
\int x^{-1} Q_{2k+1}(x) d \mu(x)=1,\ \ \ k=0,\dots,n-1.
\end{equation*}
Finally, using $Q_{2k+1}(x)=x\, q_{2k+1}(x)$ we obtain the claim.
\end{proof}
\begin{proposition}\label{prop:p2k}
\begin{subequations}
\begin{align}
&P_{2k}(z) =  \int
\frac{zQ_{2k}(z)-xQ_{2k}(x)}{x(z-x)}\,\,d\mu(x), \\
&P_{2k+1}(z) =  \int
\frac{zQ_{2k+1}(z)-xQ_{2k+1}(x)}{x(z-x)}\,\,d\mu(x).
\end{align}
\end{subequations}
\end{proposition}
\begin{proof}
From the second approximation formula of theorem \ref{thm:approx}
we have
\[
W(z)Q_{2k}(z)=P_{2k}(z)+\mathcal{O}(\frac{1}{z^{k+1}}),\ \ \ z
\rightarrow \infty.
\]
By Theorem \ref{thm:ST} $W(z)=\gamma +\int \frac{1}{z-x}d\mu(x)$
which implies
\begin{equation*}
P_{2k}(z)=\gamma Q_{2k}(z)+\int \frac{Q_{2k}(z)-Q_{2k}(x)}{z-x}d\mu(x).
\end{equation*}
Recalling that $\gamma=\int \frac1x d\mu(x)$ and performing elementary manipulations
we obtain the claim.  The proof for $P_{2k+1}$ is analogous.
\end{proof}
Let $c_j=\int x^j d\mu(x), \, j\in \mathbb Z$ be the $j$th moment
of the measure $d\mu(x)$ and let $I=[I_{ij}]=[c_{i+j}]$ be the
Hankel matrix of the moments of the measure $d\mu(x)$ appearing in
Theorem \ref{thm:ST}. We note that the shift $\gamma=c_{-1}$.  In
addition, let us denote by $\Delta^i_k$ the minor of the submatrix
of $I$ of size $k\times k$ which starts at $I_{0 i}$. By
convention, $\Delta_{0}^{j}=1, \, j\in \mathbb Z$.
\begin{theorem}\label{thm:sol-approx}
Suppose we are given an arbitrary discrete measure $d\mu(x)=\sum_j a_j \delta _{\lambda_j}, \quad \lambda_j >0,$
and define its shifted Stieltjes transform
\begin{equation*}
W(z)=\gamma+\int \frac{1}{z-x}d\mu(x),\qquad \gamma=\int \frac1x d\mu(x).
\end{equation*}
Let $k$ be a natural number such that $0\leq k\leq n-1$.  Then
there exists a unique solution to the approximation problem stated in Theorem \ref{thm:approx} whose polynomials satisfy degree
requirements dictated by Corollary \ref{cor:degrees}:

\begin{subequations}
\begin{align}
&Q_{2k}(z)=\frac{1}{\Delta_{k}^{1}}
\begin{vmatrix} 1&z&z^2&\dots&z^k\\
c_{0}&c_{1}&c_{2}&\dots&c_{k}\\
c_{1}&c_{2}&c_{3}&\dots&c_{k+1}\\
\vdots&\vdots&\vdots& &\vdots\\
c_{k-1}&c_{k}&c_{k+1}&\dots&c_{2k-1}
\end{vmatrix},\ \
\\
\notag\\
&Q_{2k+1}=\frac{1}{\Delta_{k+1}^{0}}
\begin{vmatrix} z&z^2&z^3&\dots&z^{k+1}\\
c_{1}&c_{2}&c_{3}&\dots&c_{k+1}\\
\vdots&\vdots&\vdots& &\vdots\\
c_{k}&c_{k+1}&c_{k+2}&\dots&c_{2k}
\end{vmatrix},\ \ \\
\notag\\
&P_{2k}(z) =  \int
\frac{zQ_{2k}(z)-xQ_{2k}(x)}{x(z-x)}\,\,d\mu(x), \\
\notag\\
&P_{2k+1}(z) =  \int
\frac{zQ_{2k+1}(z)-xQ_{2k+1}(x)}{x(z-x)}\,\,d\mu(x).
\end{align}
\end{subequations}

\end{theorem}
\begin{proof}
The proof is based on the work of Stieltjes \cite{stieltjes}.  We
give only a complete argument for $Q_{2k}$, the proof for
$Q_{2k+1}(z)$ is analogous.  Both proofs rely on
Proposition \ref{prop:orthogonality}.

We write $Q_{2k}(z)=\sum_{i=0}^k q_{i}z^i$.
Now, using orthogonality conditions and the fact that
$Q_{2k}[0]=1$ we have
\[
\int x^{j}(1+\sum_{i=1}^{k} q_{i}x^{i})d\mu(x)=0,\ \
j=0,1,\dotsc,k-1,\ \ \ k=0,\dots,n-1.
\]
Recall that $c_{j}=\int x^{j}d\mu(x)$. Then we have
\[
\sum_{i=1}^{k} c_{i+j}q_{i}=-c_{j},\ \ j=0,1,\dotsc,k-1.
\]
Hence we obtain the system
\[
Bq=-c,
\]
where
\[
B=\begin{pmatrix}c_{1}&c_{2}&\dots&c_{k}\\
c_{2}&c_{3}&\dots&c_{k+1}\\
\vdots&\vdots& &\vdots\\
c_{k}&c_{k+1}&\dots&c_{2k-1} \end{pmatrix},\ \
q=(q_{1},q_{2}\dotsc,q_{k})^T,\ \
c=(c_{0},c_{1},\dotsc,c_{k-1})^T.
\]
Thus, Cramer's rule implies that
\begin{equation*}
Q_{2k}(z)=\frac{1}{\Delta_{k}^{1}}
\begin{vmatrix} 1&z&z^2&\dots&z^k\\
c_{0}&c_{1}&c_{2}&\dots&c_{k}\\
c_{1}&c_{2}&c_{3}&\dots&c_{k+1}\\
\vdots&\vdots&\vdots& &\vdots\\
c_{k-1}&c_{k}&c_{k+1}&\dots&c_{2k-1}
\end{vmatrix}.
\end{equation*}
The formulas for $P_{2k}$ and $P_{2k+1}$ were already obtained in Proposition \ref{prop:p2k}.
 \end{proof}
Using Theorem \ref{thm:formulasIP} we arrive at the
inversion formulas
\begin{corollary}\label{cor:xm}
The solution to the inverse problem is given by:
\boxed{
e^{2x_{k'}}=\frac{\Delta_{k+1}^{-1}}{\Delta_k^1}, \qquad \sum_{j=0}^km_{j'}=
\frac{2}{\Delta_{k+1}^{0}}
\begin{vmatrix} c_{-1}&c_0&c_1&\dots&c_{k-1}\\
c_{1}&c_{2}&c_{3}&\dots&c_{k+1}\\
\vdots&\vdots&\vdots& &\vdots\\
c_{k}&c_{k+1}&c_{k+2}&\dots&c_{2k}
\end{vmatrix},\ \ k=0,1,\dotsc,n-1,\\}
\end{corollary}

\begin{proof}
By Theorem \ref{thm:formulasIP} $e^{x_{k'}}=P_{2k}[0]$.
From Theorem \ref{thm:sol-approx} we obtain
\begin{equation*}
P_{2k}[0]=\int \frac{Q_{2k}(x)}{x} d\mu(x),
\end{equation*}
which, in view of the determinantal formula presented in Theorem \ref{thm:sol-approx},
implies the first formula.  Likewise, to get the second formula we note
\begin{equation*}
P_{2k+1}[1]=P_{2k+1}'(0)=\int \frac{q_{2k+1}(x)}{x} d\mu(x),
\end{equation*}
and use Theorems \ref{thm:sol-approx} and \ref{thm:formulasIP} to justify the claim.

\end{proof}
\begin{remark} Although we are not going to give a complete argument here, it is easy to
show using for example the technique of Lemma 5.4 in \cite{beals-sattinger-szmigielski-moment} that
all determinants appearing in the inversion are strictly positive.  Also, it is
worth noticing that all these determinants are in fact minors of the moment matrix $I$.
\end{remark}

\begin{example}
Here we solve the inverse problem for $n=1$ and $n=2$. If $n=1$,
by corollary \ref{cor:xm} we get
\begin{equation*}
e^{2x_1}=\frac{\Delta_1^{-1}}{\Delta_0^1}=c_{-1}.
\end{equation*}
So
\begin{equation}
x_1=\frac{1}{2}\log{c_{-1}}=\frac{1}{2}\log(\frac{a_1}{\lambda_1})
\end{equation}
and
\begin{equation}
m_1=\frac{2c_{-1}}{c_0}=\frac{2}{\lambda_1}.
\end{equation}
If n=2, we have
\begin{equation*}
\begin{split}
&e^{2x_1}=\frac{\Delta_2^{-1}}{\Delta_1^1}=\frac{c_{-1}c_1-c_0^2}{c_1},\\
&e^{2x_2}=\frac{\Delta_1^{-1}}{\Delta_0^1}=c_{-1}.
\end{split}
\end{equation*}
Therefore,
\begin{equation}
\begin{split}
&x_1=\frac{1}{2}\log\frac{c_{-1}c_1-c_0^2}{c_1}=\frac{1}{2}\log\frac{a_1a_2(\lambda_1-\lambda_2)^2}{\lambda_1 \lambda_2 (a_1\lambda_1+a_2\lambda_2)}, \\
&x_2=\frac{1}{2}\log{c_{-1}}=\frac{1}{2}\log\frac{a_1\lambda_2+a_2\lambda_1}{\lambda_1\lambda_2}.
\end{split}
\end{equation}
Also, we have
\begin{equation}
\begin{split}
&m_1+m_2=\frac{2}{\Delta_2^0}(c_{-1}c_2-c_0c_1)=\frac{2(c_{-1}c_2-c_0c_1)}{c_0c_2-c_1^2},\\
&m_2=\frac{2c_{-1}}{c_0}=\frac{2(a_1\lambda_2+a_2\lambda_1)}{\lambda_1\lambda_2(a_1+a_2)}.
\end{split}
\end{equation}
Therefore,
\begin{equation}
m_1=\frac{2c_1(c_{-1}c_1-c_0^2)}{c_0(c_0c_2-c_1^2)}=\frac{2(a_1\lambda_1+a_2\lambda_2)}{\lambda_1 \lambda_2 (a_1+a_2)}.
\end{equation}
These formulas (after shifting $\lambda_j \rightarrow -\lambda_j, \, a_j \rightarrow 2a_j \lambda _j$ )
are identical to what is given on page 246 in
\cite{beals-sattinger-szmigielski-moment}.

\end{example}

\bibliographystyle{plain}
\bibliography{IP}

\begin{thebibliography}{1}

\bibitem{beals-sattinger-szmigielski-moment}
Richard Beals, David~H. Sattinger, and Jacek Szmigielski.
\newblock Multipeakons and the classical moment problem.
\newblock {\em Advances in Mathematics}, 154:229--257, 2000.

\bibitem{BGS1}
M.~Bertola, M.~Gekhtman, and J.~Szmigielski.
\newblock Cauchy biorthogonal polynomials.
\newblock {\em J. Approx. Theory}, 162(4):832--867, 2010.

\bibitem{camassa-holm}
Roberto Camassa and Darryl~D. Holm.
\newblock An integrable shallow water equation with peaked solitons.
\newblock {\em Phys. Rev. Lett.}, 71(11):1661--1664, 1993.

\bibitem{dym-mckean-gaussian}
Harry Dym and Henry~P. McKean.
\newblock {\em Gaussian processes, function theory, and the inverse spectral
  problem}.
\newblock Academic Press [Harcourt Brace Jovanovich Publishers], New York,
  1976.
\newblock Probability and Mathematical Statistics, Vol. 31.

\bibitem{gantmacher-krein}
Felix~R. Gantmacher and Mark~G. Krein.
\newblock {\em Oscillation matrices and kernels and small vibrations of
  mechanical systems}.
\newblock AMS Chelsea Publishing, Providence, RI, revised edition, 2002.
\newblock Translation based on the 1941 Russian original, edited and with a
  preface by Alex Eremenko.

\bibitem{Kreinstring}
M.~G. Krein.
\newblock On inverse problems for a nonhomogeneous cord.
\newblock {\em Doklady Akad. Nauk SSSR (N.S.)}, 82:669--672, 1952.

\bibitem{novikov}
Vladimir Novikov.
\newblock Generalizations of the {C}amassa-{H}olm equation.
\newblock {\em J. Phys. A}, 42(34):342002, 14, 2009.

\bibitem{stieltjes}
T.-J. Stieltjes.
\newblock Recherches sur les fractions continues.
\newblock {\em Ann. Fac. Sci. Toulouse Sci. Math. Sci. Phys.}, 8(4):J1--J122,
  1894.

\end{thebibliography}

\end{document}